\documentclass{llncs}

\usepackage{multicol}
\usepackage{amsmath}
\usepackage{amsfonts}
\usepackage{amssymb}
\usepackage{graphicx}
\usepackage{epic}
\usepackage{eepic}
\usepackage{epsfig,float}
\usepackage{verbatim}
\usepackage{pdfsync}

\newcommand{\bipath}{bipath}
\newcommand{\bipaths}{bipaths}

\usepackage{xcolor}

\renewcommand{\le}{\leqslant}
\renewcommand{\ge}{\geqslant}

\newcommand{\ol}{\overline}
\newcommand{\eps}{\varepsilon}
\newcommand{\emp}{\emptyset}

\newcommand{\Sig}{\Sigma}

\newcommand{\noin}{\noindent}

\newcommand{\bi}{\begin{itemize}}
\newcommand{\ei}{\end{itemize}}
\newcommand{\be}{\begin{enumerate}}
\newcommand{\ee}{\end{enumerate}}
\newcommand{\bd}{\begin{description}}
\newcommand{\ed}{\end{description}}
\newcommand{\bq}{\begin{quote}}
\newcommand{\eq}{\end{quote}}

\newcommand{\tid}{\mbox{{\bf 1}}}
\newcommand{\cA}{{\mathcal A}}
\newcommand{\cB}{{\mathcal B}}
\newcommand{\cC}{{\mathcal C}}
\newcommand{\cD}{{\mathcal D}}

\newcommand{\cN}{{\mathcal N}}

\newcommand{\cT}{{\mathcal T}}

\newcommand{\gR}{{\mathcal R}}
\newcommand{\gJ}{{\mathcal J}}

\newcommand{\one}{{\mathbf 1}}
\newcommand{\Lra}{{\hspace{.1cm}\Leftrightarrow\hspace{.1cm}}}

\newcommand{\qedb}{\hfill$\blacksquare$}

\title{Large Aperiodic Semigroups\thanks{This work was supported by the Natural Sciences and Engineering Research Council of Canada 
grant No.~OGP000087
and by Polish NCN grant DEC-2013/09/N/ST6/01194.}
}

\author{Janusz~Brzozowski\inst{1} \and Marek Szyku{\l}a \inst{2}}

\titlerunning{Large Aperiodic Semigroups}

\authorrunning{J. Brzozowski and M. Szyku{\l}a}   

\institute{David R. Cheriton School of Computer Science, University of Waterloo, \\
Waterloo, ON, Canada N2L 3G1\\
\{{\tt brzozo@uwaterloo.ca}\}
\and
Institute of Computer Science, University of Wroc{\l}aw,\\
Joliot-Curie 15, PL-50-383 Wroc{\l}aw, Poland\\
\{{\tt msz@cs.uni.wroc.pl}\}
}

\begin{document}
\maketitle
\begin{abstract}
The syntactic complexity of a regular language is the size of its syntactic semigroup.
This semigroup  is isomorphic to the transition semigroup of a minimal deterministic finite automaton accepting the language, that is, to the semigroup generated by transformations induced by non-empty words on the set of states of the automaton. 
In this paper we search for the largest syntactic semigroup of a star-free language having $n$ left quotients; equivalently, we look for the largest transition semigroup of an aperiodic finite automaton with $n$ states.

We introduce two new aperiodic transition semigroups. The first  is generated by transformations that change only one state; we call such transformations and resulting semigroups \emph{unitary}. 
In particular, we study \emph{complete} unitary semigroups which have a special structure, and we show that each maximal unitary semigroup is complete. For $n \ge 4$ there exists a complete unitary semigroup that is larger than any aperiodic semigroup known to date.

We then present even larger aperiodic semigroups, generated by transformations that map a non-empty subset of states to a single state; we call such transformations and semigroups \emph{semiconstant}.
In particular, we examine semiconstant \emph{tree} semigroups which have a structure based on full binary trees. 
The semiconstant tree semigroups are at present the best candidates for largest aperiodic semigroups.

We also prove that $2^n-1$ is an upper bound on the state complexity of reversal of star-free languages, and resolve an open problem about a special case of state complexity of concatenation of star-free languages.
\smallskip

\noin
{\bf Keywords:}
aperiodic, monotonic, nearly monotonic, partially monotonic, semiconstant, transition semigroup, star-free language, syntactic complexity, unitary
\end{abstract}

\section{Introduction}
The \emph{state complexity} of a regular language is the number of states in a complete minimal deterministic finite automaton (DFA) accepting the language~\cite{Yu01}. An equivalent notion is that of \emph{quotient complexity,} which is the number of left quotients of the language~\cite{Brz10}; we prefer quotient complexity since it is a language-theoretic notion.
The usual measure of complexity of an operation on regular languages~\cite{Brz10,Yu01} is the quotient complexity of the result of the operation as a function of the quotient complexities of the operands.
This measure has some serious disadvantages, however. For example, as shown in~\cite{BrLiu12}, in the class of star-free languages all common operations have the same quotient complexity as they do in the class of arbitrary regular languages\footnote{Two small exceptions are discussed in Section~\ref{sec:rev}.}. Thus quotient complexity fails to differentiate between the very special class of star-free languages and the class of all regular languages.

It has been suggested that other measures of complexity may also be useful~\cite{Brz12a}, 
in particular, 
the \emph{syntactic complexity} of a regular language which is the cardinality of its syntactic semigroup~\cite{Pin97}. This is the same as the cardinality of the \emph{transition semigroup} of a minimal DFA accepting the language, and it is this latter representation that we use here. 
The transition semigroup is the set of all transformations induced by non-empty words on the set of states of the DFA. 
The \emph{syntactic complexity of a class} of languages is the size of the largest syntactic semigroups of languages in that class as a function of the quotient complexities of the languages.
Since the syntactic complexity of star-free languages is considerably smaller than that of regular languages,  this measure succeeds in distinguishing the two classes.

The class of \emph{star-free} languages is the smallest class  obtained from finite languages using only boolean operations and concatenation, but no star.
By Sch\"utzenberger's theorem~\cite{Sch65} we know that a language is star-free if and only if the transition semigroup of its minimal DFA is \emph{aperiodic}, meaning that it contains no non-trivial subgroups.
Equivalently, a transition semigroup is aperiodic if and only if no word over the alphabet of the DFA can induce  a non-trivial permutation of any subset of two or more states.
Star-free languages and the DFAs that accept them were studied by McNaughton and Papert in 1971~\cite{McNP71}.

Two aperiodic semigroups, monotonic and partially monotonic, were studied by Gomes and Howie~\cite{GoHo92}.
Their results were adapted to finite automata in~\cite{BLL13},  where nearly monotonic semigroups were also introduced; they are larger than the partially monotonic ones and were the largest aperiodic semigroups known to date for $n \le 7$. For $n \ge 8$ the largest aperiodic semigroups known to date were  those generated by DFAs accepting $\mathcal{R}$-trivial languages~\cite{BrLi13}. The syntactic complexity of $\mathcal{R}$-trivial languages is $n!$.
As to aperiodic semigroups, tight upper bounds on their size were known only for $n\le3$.
\smallskip

The following are the main contributions of this paper:
\be
\item
Using the method of~\cite{KiSz13}, we have
enumerated all aperiodic semigroups for $n=4$, and we have shown that the maximal aperiodic semigroup has size 47, while the maximal nearly monotonic semigroup has size 41.
Although this may seem like an insignificant result, it provided us with strong motivation to search for larger semigroups.

The number of aperiodic transformations is $(n+1)^{n-1}$.
For large $n$ the number of aperiodic semigroups is very large, and so it is difficult to check them all.

\item
We studied  semigroups  generated by transformations that change only one state; we call such transformations and semigroups \emph{unitary}.
We characterized unitary semigroups and computed their maximal sizes up to $n=1,000$.
There are $n(n-1)$ unitary transformations. For $n\ge 4$ the maximal unitary semigroups are larger than the maximal nearly monotonic ones and also  larger than  
any previously known aperiodic semigroup.
\item
For each $n$ we found a set of DFAs whose inputs induce \emph{semiconstant tree} transformations --
transformations that send a non-empty subset of the set of all states to a single state, and have a structure based on full binary trees. 
For $n\ge 4$, there is a semiconstant tree semigroup larger than the largest complete unitary semigroup.
We computed the maximal size of these transition semigroups up to $n=500$.
The total number of semiconstant transformations is $(2^{n-1}-1)n$.
\item
We  derived formulas for the sizes  of complete unitary and semiconstant tree semigroups. We also provided recursive formulas characterizing the maximal complete unitary and semiconstant tree semigroups; these formulas lead to efficient algorithms for computing the forms and sizes of such semigroups.
\item
We proved that the quotient complexity of the reverse of a star-free language with quotient complexity $n$ is at most $2^n-1$.
\item
 We resolved an open problem about a special case of quotient complexity of product (catenation, concatenation) of star-free languages $K$ and $L$, when the quotient complexities of $K$ and $L$ are $m\ge 2$ and $2$, respectively: we proved that $3m-2$ is a tight upper bound. 
\ee
Our results about aperiodic semigroups are summarized in 
Tables~\ref{tab:aperiodic_summary} and~\ref{tab:aperiodic_summary2} for small values of $n$. Transformation $\one$ is the identity; it can be added to unitary and semiconstant transformations without affecting aperiodicity.

Additional information about the classes of semigroups in Tables~\ref{tab:aperiodic_summary} and~\ref{tab:aperiodic_summary2} will be given later. 
The classes are listed in the order of increasing size when $n$ is large.
The number in boldface shows the value of $n$ for which the size of a given semigroup exceeds the sizes of all of the preceding ones.
For example, the largest semigroups of finite languages exceed the preceding semigroups for $n\ge 12$.

There are two more classes of syntactic semigroups that have the same complexity as the semigroups of finite languages: those of cofinite and reverse definite languages.
The tight upper bound $\lfloor e$ $\cdot$ $ (n-1)! \rfloor$ for $\mathcal{J}$-trivial languages (\cite{BrLi13}) is also a lower bound for definite languages (\cite{BLL13}).
An upper bound of $n((n-1)!-(n-3)!)$ has been shown to hold~\cite{IvNG14} for definite and generalized definite languages~\cite{Gin66}, but it is not known whether this bound is tight.

The asymptotic behaviour of the size $e(n)$ of partially monotonic semigroups is 
$A\frac{B^{2n-1}} {\sqrt{n}}$, 
where $A$ and $B$ are constants~\cite{BLL13}. 
For nearly monotonic semigroups the size is $e(n)+n-1$.

\renewcommand{\arraystretch}{1.3}
\begin{table}
\caption{ Large aperiodic semigroups.}
\label{tab:aperiodic_summary}
\begin{center}
$
\begin{array}{|l||c|c|c|c|c|c|c|c|}    
\hline
\hfill  n\colon &
\ \ \ 1 \ \ \  & \ \ \ 2 \ \ \  &  \ \ \ 3 \ \ \  & \ \ \ 4 \ \ \ &\ \ \ 5 \ \ \ &  6 &  7  &  \ \ \ 8 \ \ \ \\
\hline \hline
\ \text{Monotonic }  \binom{2n-1}{n} &	
{ 1}	&	{ 3}	& { 10} &	35	& 126  &  \ \ \ 462 &  \ \ \ 1,716   &\ \ \ 6,435   \\
\hline
\ \text {Part. mon. } e(n)
 &	-				&	2					&  \ 8					& \bf	38	
 &	192  &\  1,002 &\ \ \  5,336 &\ \ 28,814\\
\hline
\ \text {Near. mon. } e(n)+n-1
&	-				&	{ 3}				& {\bf 10}			&	41
&	196  &\  1,007 &\ \ \    5,342 & \ \ 28,821 \\
\hline
\ \text {Finite } (n-1)!
&	1				&	1
				& 2			&	6
&	24  &\  120 &\ \ \    720 & \ \ 5,040 \\
\hline
\ \mathcal{J}\text{-trivial } \lfloor e$ $\cdot$ $ (n-1)! \rfloor
&	1				&	2
				& 5			&	16
&	65  &\  326 &\ \ \    1,957 & \ \ 13,700 \\
\hline
\ \mathcal{R}\text{-trivial } n!
&	  1				&	2			& 	6					&    24	
&	 120	&\ \ 720 &\ \ \  5,040 &\ \  \bf 40,320\\
\hline
\ \text{Complete unitary with {\bf 1}}
&	  		-		&	3			& 	10					&   \bf  45	
&	 270	&\  1,737 &\   13,280 &\  121,500 \\
\hline
\ \text{Semiconstant tree with {\bf 1} }
&	  -				&	3			& 	10					&  \bf  47	
&	 273	& \  1,849 & \ 14,270 & \ 126,123 \\

\hline
\ \text{Aperiodic }
&	  1				&	3			& 	10					&   	47	
&	?	& ? & ? & ?\\
\hline
\end{array}
$
\end{center}
\end{table}
\newpage

\begin{table}[h]
\caption{ Large aperiodic semigroups continued.}
\label{tab:aperiodic_summary2}
\begin{center}
$
\begin{array}{|l||c|c|c|c|c|}    
\hline
\hfill  n\colon &
\ \ \ 9 \ \ \  & \ \ \ 10 \ \ \  &  \ \ \ 11 \ \ \  & \ \ \ 12 \ \ \ &\ \ \ 13 \ \ \ \\
\hline \hline
\ \text{Monotonic}  &	24,310
	& 92,378		& 352,716  &	1,352,078	& 5,200,300  \\
\hline
\ \text {Part. mon. }
 &	157,184				&	864,146					&  	4,780,008			& 		
 26,572,086 &	148,321,344  \\
\hline
\ \text {Near. mon. }
&	157,192				&	864,155				& 	4,780,018		&	26,572,097
& 148,321,352	   \\
\hline
\ \text {Finite}
&	40,320				&	362,880
		 & 	3,628,800 &	\bf 39,916,800
&	479,001,600  \\
\hline
\  \mathcal{J}\text{-trivial }
& 109,601				&	\bf 986,410
				& 	9,864,101	& 108,505,112	  &	1,302,061,345  \\
\hline
\ \mathcal{R}\text{-trivial }
&	  362,880				&	3,628,800			& 	39,916,800					&    479,001,600	&	 6,227,020,800	 \\
\hline
\ \text{Comp. unit., {\bf 1}}
&	  1,231,200				&	12,994,020			& 	151,817,274				&    2,041,564,499	
&	 29,351,808,000	 \\
\hline
\ \text{Sc. tree, {\bf 1} }
&	  1,269,116				&		14,001,630		& 	169,410,933				&   2,224,759,334	
&	31,405,982,420	  \\

\hline
\ \text{Aperiodic }
&	  ?				&	?			& 	?					&   	?	
&	?	 \\
\hline
\end{array}
$
\end{center}
\end{table}

The remainder of the paper is structured as follows. Section~\ref{sec:term} presents our terminology and notation. 
Our large aperiodic semigroups are defined in Section~\ref{sec:semi}.
The special case of unitary semigroups is then considered in Section~\ref{sec:unitary}, and semiconstant tree semigroups are the topic of Section~\ref{sec:semi2}.
Section~\ref{sec:rev} contains the new results about reversal and product.
Section~\ref{sec:conc} concludes the paper.

\section{Terminology and Notation}
\label{sec:term}

Let $\Sig$ be a finite alphabet. The elements of $\Sig$ are \emph{letters} and the elements of $\Sig^*$ are \emph{words}, where $\Sig^*$ is the free monoid generated by $\Sig$.
The empty word is denoted by $\eps$, and the set of all non-empty words is $\Sig^+$, the free semigroup generated by $\Sig$. A \emph{language} is any subset of $\Sig^*$.

Suppose $n\ge 1$. Without loss of generality we assume that our basic set under consideration is
$Q=\{0,1,\dots,n-1\}$. A~\emph{deterministic finite automaton (DFA)} is a quintuple $\cD=(Q, \Sig, \delta, 0,F)$, where $Q$ is a finite non-empty set of \emph{states}, $\Sig$ is a finite non-empty \emph{alphabet}, $\delta\colon Q\times \Sig\to Q$ is the \emph{transition function}, $0\in Q$ is the \emph{initial state}, and $F\subseteq Q$ is the set of \emph{final states}. We extend $\delta$ to $Q \times \Sig^*$ and to $2^Q\times \Sig^*$ in the usual way.
A DFA $\cD$ \emph{accepts} a word $w \in \Sigma^*$ if $\delta(0,w)\in F$. 
The \emph{language accepted} by  $\cD$ is $L(\cD)=\{w\in\Sig^*\mid \delta(0,w)\in F\}$. 

By the \emph{language of a state} $q$ of $\cD$ we mean the language $L_q(\cD)$ accepted by the DFA $(Q,\Sigma,\delta,q,F)$. A state is \emph{empty} (also called \emph{dead} or a \emph{sink}) if its language is empty. 
Two states $p$ and $q$ of $\cD$ are \emph{equivalent} if $L_p(\cD) = L_q(\cD)$. 
Otherwise, 
states $p$ and $q$ are \emph{distinguishable}. 
A state $q$ is \emph{reachable} if there exists a word $w\in\Sig^*$ such that $\delta(0,w)=q$.
A DFA is \emph{minimal} if all its states are reachable and pairwise distinguishable.

A \emph{transformation} of $Q$ is a mapping of $Q$ into itself. Let $t$ be a transformation of $Q$; then $qt$ is the \emph{image} of $q\in Q$ under $t$. If $P$ is a subset of $Q$, then $Pt = \{qt \mid q \in P\}$.
An arbitrary transformation can be written in the form
\begin{equation*}
t=\left( \begin{array}{ccccc}
0 & 1 &   \cdots &  n-2 & n-1 \\
p_0 & p_1 &   \cdots &  p_{n-2} & p_{n-1}
\end{array} \right ),
\end{equation*}
where $p_q = qt$  for $q\in Q$.
We also use $t = [p_0,\ldots,p_{n-1}]$ as a simplified notation. The {\em composition} of two transformations $t_1$ and $t_2$ of $Q$ is a transformation $t_1 \circ t_2$ such that $q (t_1 \circ t_2) = (q t_1) t_2$ for all $q \in Q$. We usually drop the composition operator ``$\circ$'' and write $t_1t_2$.

Let $\cT_{Q}$ be the set of all $n^n$ transformations of $Q$; then $\cT_{Q}$ is a monoid under  composition. 
The \emph{identity} transformation $\tid$ maps each element to itself, that is, $q\tid=q$ for all $q \in Q$.
 For $k\ge 2$, a transformation (permutation) $t$ of a set $P=\{q_0,q_1,\ldots,q_{k-1}\} \subseteq Q$ is a \emph{$k$-cycle}
if $q_0t=q_1, q_1t=q_2,\ldots,q_{k-2}t=q_{k-1},q_{k-1}t=q_0$.
A $k$-cycle is denoted by $(q_0,q_1,\ldots,q_{k-1})$.
If a transformation $t$ of $Q$ acts like a $k$-cycle on some $P \subseteq Q$, we say that $t$ has a $k$-cycle.
A~transformation has a \emph{cycle} if it has a $k$-cycle for some $k\ge 2$.
For $p\neq q$, a \emph{transposition} is the 2-cycle $(p,q)$.
A~\emph{permutation} of $Q$ is a mapping of $Q$ \emph{onto} itself. 
A transformation is \emph{aperiodic} if it contains no cycles.

In any DFA $\cD$, each word $w\in\Sig^*$ induces a transformation $t_w$ of $Q$ defined by $qt_w=\delta(q,w)$ for all $q\in Q$.
The set of all transformations of $Q$ induced in $\cD$ by non-empty words is the \emph{transition semigroup} of $\cD$.
This semigroup is a subsemigroup of $\cT_Q$.
If $\cD$ is minimal, its transition semigroup is isomorphic to the \emph{syntactic semigroup} of the language $L(\cD)$~\cite{McNP71,Pin97}.  
A language is regular if and only if its syntactic semigroup is finite. The size of the syntactic semigroup of a language is called its \emph{syntactic complexity}.
In this paper we deal only with transition semigroups; consequently, we view syntactic complexity as the size of the transition semigroup. 

If  $T$ is a set of transformations,  then $\langle T \rangle$ is the semigroup generated by $T$. 
If $\cD=(Q,\Sig,\delta,0,F)$ is a DFA, the  transformations induced by letters of $\Sig$ are called  \emph{generators of the transition semigroup} of $\cD$ or simply \emph{generators} of $\cD$. 

\section{Unitary and Semiconstant DFAs}
\label{sec:semi}
We now define a new class of aperiodic DFAs among which are found the largest transition semigroups known to date. We also study several of its subclasses.

A \emph{unitary} transformation $t$, denoted by $(p\to q)$, has $p\neq q$, $pt=q$ and $rt=r$ for all $r\neq p$. 
A~DFA is \emph{unitary} if each of its generators is unitary.
A semigroup is \emph{unitary} if it has a set of unitary generators.

A~\emph{constant} transformation $t$,  denoted by $(Q \to q)$, has $pt=q$ for all $p\in Q$.
A~transformation $t$  is \emph{semiconstant} if it maps a non-empty subset $P$ of $Q$ to a single element $q$ and leaves the remaining elements of $Q$ unchanged. It is denoted by $(P\to q)$. A constant transformation is semiconstant with $P=Q$,  and a unitary transformation $(p\to q)$ is semiconstant with $P=\{p\}$ (or $P=\{p,q\}$).
A~DFA is \emph{semiconstant} if each of its generators is semiconstant.
A~semigroup is \emph{semiconstant} if it has a set of  semiconstant generators.


For each  $n\ge 1$ we shall define several DFAs.
Let $m$, $n_1,n_2,\dots,n_m$ be positive natural numbers.
Also,  let $n=n_1+\dots+n_m$, and for each $i$, $1\le i \le m$, define $r_i$ by 
$r_i= \sum_{j=1}^{i-1} n_j$.
For $i=1,\dots, m$, let $Q_i=\{r_i,r_i+1,\dots, r_{i+1}-1\}$; thus the cardinality of $Q_i$ is $n_i$.
Let $Q=Q_1\cup \dots\cup Q_m=\{0,\dots,n-1\}$; the cardinality of $Q$ is $n$.
The sequence $(n_1,n_2,\dots,n_m)$ is called the \emph{distribution} of $Q$.
\begin{remark}
\label{rem:dist}
The number $d(n)$ of   different distributions for each $n$ is $2^{n-1}$.
This is easily verified by induction on $n$. For $n=1$ there is only one distribution, namely $(1)$; hence $d(1)=2^0$. Suppose that $d(k)=2^{k-1}$ for $k<n$. For $n$, each distribution 
is either $(n)$ or it has  $m$, where $m\in \{n-1,\dots,1\}$, combined with any distribution of the integer $n-m$. 
Hence the number of distributions is
$1 +d(1)+d(2)+\dots +d(n-1)=1+1+2+\dots +2^{n-2}=2^{n-1}.$
For example, for $n=3$ we have the distributions $(3)$, $(2,1)$, $(1,2)$, $(1,1,1)$.
\qedb
\end{remark}
A binary tree is \emph{full} if every vertex has either two children or no children. 
There are $C_{m-1}$ full binary trees, where $C_m=\frac{1}{m+1}\binom{2n}{n}$
is the Catalan number\footnote{http://en.wikipedia.org/wiki/${\rm Catalan}\_{}{\rm number}$}.

Let $\Delta_Q$ be a full binary tree with $m$ leaves labeled $Q_1,\dots,Q_m$ from left to right.  To each node $v \in \Delta_Q$,  we assign  the union $\mathrm{Q}(v)$ of all the sets $Q_i$ labeling the leaves in 
the subtree rooted at~$v$.

With each full binary tree we can associate different distributions.
A full binary tree $\Delta_Q$ with a distribution attached is denoted by $\Delta_Q(n_1,n_2,\dots,n_m)$ and is called the \emph{structure} of $Q$.
This structure will uniquely determine the transition function $\delta$ of the DFAs defined below.
The number of possible structures of $Q$ for a given $n$ is the binomial transform of $C_n$, the Catalan number\footnote{\tt http://oeis.org/A007317}. 

We can denote the structure of $Q$ as a binary expression. For example, the expression $((3,2),(4,1))$ denotes the full binary tree in which the leaves are labeled $Q_1$, $Q_2$, $Q_3$, and $Q_4$, where 
$|Q_1|=3, |Q_2|=2, |Q_3|=4, |Q_4|=1$, and the interior nodes are labeled by $Q_1\cup Q_2$, $Q_3\cup Q_4$ and $Q_1\cup Q_2 \cup Q_3\cup Q_4$. 
On the other hand, the expression $(((3,2),4),1)$ has interior nodes labeled 
$Q_1\cup Q_2$, $Q_1\cup Q_2 \cup Q_3$ and $Q_1\cup Q_2 \cup Q_3\cup Q_4$.

\begin{definition}[Transformations]
\label{def:semi}
\bd
\item[Type 1:]
Suppose $n> 1$ and $(n_1,n_2,\dots,n_m)$ is a distribution of $Q$. For
all $i=1,\dots,m$ and $q, q+1\in Q_i$ Type 1 transformations are the unitary transformations $(q\to q+1)$ and $(q+1\to q)$. 
\item[Type 2:]
Suppose $n> 1$ and $(n_1,n_2,\dots,n_m)$ is a distribution of $Q$.
If $1\le i\le m-1$  and $i<j\le m$, for each $q\in Q_i$ and $p\in Q_j$,
$(q\to p)$ is a Type 2
transformation.
\item[Type 3:]
Suppose $n> 1$ and $\Delta_Q(n_1,n_2,\dots,n_m)$ is a structure of $Q$.
For each internal node $w$ the semiconstant transformation
$(\mathrm{Q}(w) \to \min(\mathrm{Q}(w)))$
is of Type 3.
\item[Type 4:]
The identity transformation $\one$ on $Q$ is of Type 4.
\ed
\end{definition}

For a fixed $i$ there are $2n_i-2$ Type~1 transformations and $n_i(n_{i+1}+\dots+n_m)$ Type~2 transformations. The number of Type~3 transformations is $m-1$.

Note that the distribution $(n_1,n_2,\dots,n_m)$ affects transformations of Types 1, 2, and 3, whereas the binary tree affects only transformations of Type 3.

In the following DFAs the transition function is defined by a set of transformations and the alphabet consists of letters inducing these transformation.

\begin{definition}[DFAs]\label{def:scDFA}
Suppose $n >1$.
\be
\item
If there is no $i\in \{1,\dots, m-1\}$ such that $|Q_i|=|Q_{i+1}|=1$, 
then any DFA of the form
$\cD_{u}(n_1,\dots,n_m)=(Q,\Sig_u,\delta_{u},0,\{n-1\})$, where $\delta_u$ has all the  trans\-for\-mations of Types 1 and 2, is a \emph{complete unitary} DFA. 
\item
$\cD_{ui}(n_1,\dots,n_m)=(Q,\Sig_{ui},\delta_{ui},0,\{n-1\})$ is $\cD_{u}(n_1,\dots,n_m)$ with $\one$ added.
\item
Any DFA
$\cD_{sct}(\Delta_Q(n_1,\dots,n_m))=(Q,\Sig_{sct},\delta_{sct},0,\{n-1\})$, where $\delta_{sct}$ has all the transformations of Types 1, 2 and 3, is a \emph{semiconstant tree} DFA.
\item
$\cD_{scti}(\Delta_Q(n_1,\dots,n_m))=(Q,\Sig_{scti},\delta_{scti},0,\{n-1\})$ is $\cD_{sct}(\Delta_Q(n_1,\dots,n_m))$ with $\one$ added.
\ee
\end{definition}

Using terminology analogous to that of~\cite{Die10}, we define a \emph{bipath (bidirectional path)}
to be a graph $(V,E)$, where $V=\{v_0,\dots,v_{k-1}\}$ for some $k\ge 1$, and for each $v_q, v_{q+1}\in V$ there are two edges 
$(v_q, v_{q+1})$ and $(v_{q+1},v_q)$. 
If $k=1$, the graph $(\{v_0\}, \emp)$ is also considered a (trivial) \bipath.
If we ignore self-loops, each edge in the graph uniquely determines a unitary transformation, 
and  the states in each $Q_i$ in $\cD_{u}(n_1,\dots,n_m)$ constitute a \emph{\bipath}.
 Also, the  graph of $\cD_{u}(n_1,\dots,n_m)$ is a sequence 
$(Q_1,\dots,Q_m)$ of \bipaths, where there are transitions from every $q$ in $Q_i$ to every  $p$ in $Q_j$, if $i<j$.

\begin{example}
\label{ex:semi}
Figure~\ref{fig:semi} shows three examples of unitary DFAs. 
In Fig.~\ref{fig:semi}~(a) we have DFA $\cD_u(3)$, where the letter $a_{pq}$ induces the unitary transformation $(p\to q)$.
In Fig.~\ref{fig:semi}~(b) we present $\cD_u(3)$, where only the transitions between \emph{different} states are included to simplify the figure. Also,  the letter labels are deleted because they are easily deduced.
Next, in Figs.~\ref{fig:semi}~(c) and~(d), we have the DFAs $\cD_u(3,1)$ and $\cD_u(2,2,2)$,  respectively. We shall return to these examples later.
\qedb
\end{example}

\begin{figure}[h]
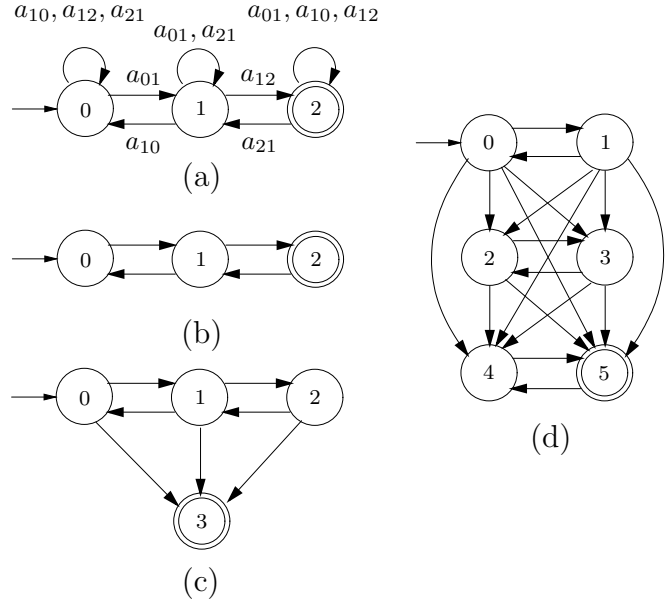

\begin{center}
\input semi.eepic
\end{center}
\caption{Unitary DFAs: (a) $\cD_u(3)$; (b) $\cD_u(3)$ simplified; 
(c) $\cD_u(3,1)$; (d)  $\cD_u(2,2,2)$. } 
\label{fig:semi}
\end{figure}

\begin{remark}
\label{rem:minimal}
All four DFAs of Definition~\ref{def:scDFA} are minimal as is easily verified. Hence the syntactic semigroup of the language of each DFA is isomorphic to the transition semigroup of the DFA.
\qedb
\end{remark}

\section{Unitary Semigroups}
\label{sec:unitary}

We study unitary semigroups because their generators are the simplest.
We begin with three previously studied special semigroups which are subsemigroups of a unitary semigroups. 

\subsection{Monotonic Semigroups}
\label{ssec:mon}
Monotonic semigroups were previously studied in~\cite{BLL13,GoHo92,How71}.
A transformation $t$ of $Q$ is \emph{monotonic} if there exists a total order $\le$ on $Q$ such that, for all $p,q \in Q$, $p \le q$ implies $pt \le qt$. 
Note that the identity transformation is monotonic.
A~DFA is \emph{monotonic} if each of its input transformations is monotonic.
A semigroup is \emph{monotonic} if it has a set of  monotonic generators.
From now on we assume that $\le$ is the usual order on integers.

The following result of~\cite{GoHo92} is somewhat modified for our purposes:
\begin{proposition}[Gomes and Howie]
\label{prop:GoHo}
The set $M$ of all $\binom{2n-1}{n}-1$ monotonic transformations other than $\one$ is an aperiodic semigroup generated by 
$$G_M=\{(q \to q+1)\mid 0\le q\le n-2\} \cup \{(q \to q-1)\mid 1\le q\le n-1\},$$
and no smaller set of unitary transformations generates $M$.
\end{proposition}
\begin{corollary}
\label{cor:mon}
The transition semigroup of $\cD_{ui}(n)$ is the semigroup $M\cup \{\one\}$ of all monotonic transformations.
\end{corollary}

Note that $\cD_{u}$ has transitions of Type 1 only, and $D_{ui}$ has  Type 1 and~4 only. 
Figure~\ref{fig:semi}~(b) shows $\cD_{u}(3)$ and $\cD_{ui}(3)$, if $\one$ is added. The transition semigroup of $\cD_{ui}(3)$ has ten elements and is the largest aperiodic semigroup for $n=3$~\cite{BLL13}.

 Note also that there are monotonic semigroups that do not have unitary generating sets;  each monotonic semigroup, however, is a subsemigroup of the transition semigroup of $\cD_{ui}(n)$ consisting of all monotonic transformations.

\subsection{Partially Monotonic Semigroups}
\label{ssec:pmon}

A \emph{partial transformation} $t$ of $Q$ is a partial mapping of $Q$ into itself.
If $t$ is defined for $q\in Q$, then $qt$ is the image of $q$ under $t$; otherwise, we write $qt=\Box$.
By convention, $\Box t=\Box$.
The \emph{domain} of $t$ is  the set $dom(t)=\{q\in Q \mid qt\neq \Box\}$.
A~partial transformation is \emph{monotonic} if there exists an order $\le$ on $Q$ such that for all $p,q\in dom(t)$,  $p\le q$ implies $pt \le qt$.

Semigroups of monotonic partial transformations were studied by Gomes and Howie~\cite{GoHo92}.  They were adapted to automata in~\cite{BLL13}. 
We follow~\cite{BLL13} by starting with all partial transformations of 
$Q\setminus \{n-1\}$  and adding state $(n-1)$  for the undefined value $\Box$.
We call the resulting transformations \emph{partially monotonic}.
The following is an adaptation of the results of~\cite{GoHo92}:
 \begin{proposition}
 \label{prop:pmon}
 For  $n\ge 2$, the DFA $\cD_{ui}(n-1,1)=(Q,\Sig_{ui},\delta_{ui},0,\{n-1\})$ has the following properties:
 \be
 \item 
 Each of the $3n-4$ transformations of $\cD_{ui}(n-1,1)$ is partially monotonic.
 Thus $\cD_{ui}(n-1,1)$ is partially monotonic, and hence aperiodic.
 \item
 The transition semigroup $PM_Q$ of $\cD_{ui}(n-1,1)$ consists of all the $e(n)$ partially monotonic transformations of $Q$, where
 \begin{equation}
 e(n)=\sum_{k=0}^{n-1} \binom{n-1}{k} \binom{n+k-2}{k}.
 \end{equation}
 \item
 Each generator is idempotent, and $3n-4$ is the smallest number of idempotent  generators of $PM_Q$. Moreover, each generator except $\one$ is unitary, and $3n-5$ is the smallest number of unitary generators of $PM_Q\setminus \{\one\}$.
 \ee
 \end{proposition}
\begin{example}
\label{ex:pmon}
There are eight monotonic partial transformations of the set $Q=\{0,1\}$, namely:
$[\Box,\Box]$, $[0,\Box]$, $[1,\Box]$, $[\Box,0]$, $[\Box,1]$, $[0,0]$, $[0,1]$, $[1,1]$.
When we replace $\Box$ by state 2, the eight partial transformations become 
total transformations
$[2,2,2]$, $[0,2,2]$, $[1,2,2]$, $[2,0,2]$, $[2,1,2]$, $[0,0,2]$, $[0,1,2]$, $[1,1,2]$.
The $9-4=5$ generators of $\cD_{ui}(2,1)$ are: $(0\to 1)=[1,1,2]$, $(0\to 2)=[2,1,2]$, $(1\to 0)=[0,0,2]$, $(1\to 2)=[0,2,2]$ and 
$\one$.
The DFA of Figure~\ref{fig:semi}~(c) is an example of  $\cD_{ui}(3,1)$.
\qedb
\end{example}
For $n\ge 4$ the semigroup of all partially monotonic transformations is larger than the semigroup of all monotonic transformations.

 Note that there are partially monotonic semigroups that do not have unitary generating sets;  each partially monotonic semigroup, however, is a subsemigroup of the transition semigroup of $\cD_{ui}(n-1,1)$ consisting of all partially monotonic transformations.

\subsection{Other Previously Studied Aperiodic Semigroups}

 As we have mentioned in the introduction, the syntactic complexity of five other language classes was studied previously. \emph{Cofinite} languages are complements of finite languages, and therefore their minimal DFAs have the same transition semigroup as the DFAs of finite languages.

The reverse $w^R$ of a word $w\in \Sig^+$ is $w$ spelled backwards and $\eps^R=\eps$.
The reverse of a language $L$ is $L^R=\{w^R\mid w\in L\}$.
A language is \emph{definite} if it has the form $E\cup \Sig^* F$, where $E$ and $F$ are finite.
It is \emph{reverse definite} if its reverse is definite, that is, if it has the form  $E\cup F\Sig^*$, where $E$ and $F$ are finite. It was shown in~\cite{BLL13} that the syntactic complexity of reverse definite languages is the same as that of finite languages. A lower bound of $\lfloor e(n-1)!\rfloor$ was proved for definite languages; it is an open question whether this is also an upper bound.

The  well known Green relations define $\mathcal{R}$-trivial and $\mathcal{J}$-trivial monoids (semigroups with an identity).
If $M$ is a monoid, the relation $\mathbin{\gR}$ is defined by 
$ s \mathbin{\gR} t \mathbin{\Lra} Ms = Mt$ for  $s,t\in M$. A monoid is  \emph{$\mathcal{R}$-trivial} if $ s \mathbin{\gR} t$ implies $s=t$.
The relation $\mathbin{\gJ}$ is defined by $  s \mathbin{\gJ} t \mathbin{\Lra} MsM = MtM, $
and $M$ is \emph{$\mathcal{J}$-trivial} if $MsM=MtM$ implies $s=t$.
Languages whose minimal DFAs have $\mathcal{R}$-trivial ($\mathcal{J}$-trivial) transition monoids are also called \emph{$\mathcal{R}$-trivial} (\emph{$\mathcal{J}$-trivial}).

Syntactic complexities of $\mathcal{R}$-trivial and $\mathcal{J}$-trivial languages were studied by Brzozowski and Li~\cite{BrLi13}. Consider the natural order $<$ on $Q$. We say that a transformation $t$ is non-decreasing if  $q \le qt$ for all $q \in Q$.  Let $\mathcal{F}_Q$ be the set of all non-decreasing transformations. The size of $\mathcal{F}_Q$ is $n!$.

It was shown in~\cite{BrFi80} that $L$ is an $\mathcal{R}$-trivial language if and only if its  minimal  DFA is partially ordered,  or equivalently, if its transition semigroup contains only non-decreasing transformations.
Thus the largest semigroup generated by DFAs accepting $\mathcal{R}$-trivial languages is $\mathcal{F}_Q$.

\begin{proposition}
The transition semigroup of $\cD_{ui}(1,1,\ldots,1)$ is the semigroup $\mathcal{F}_Q$ of all non-decreasing transformations.
\end{proposition}
\begin{proof}
DFA 
$\cD_{ui}(1,1,\ldots,1)$ has only unitary transformations of Type~2. 
They generate only non-decreasing transformations, since each of them preserves  the natural order. An arbitrary non-decreasing transformation has the form
\begin{equation*}
t=\left( \begin{array}{ccccc}
0 & 1 &   \cdots &  n-2 & n-1 \\
p_0 & p_1 &   \cdots &  p_{n-2} & {n-1}
\end{array} \right ),
\end{equation*}
where $p_q\ge q$ for $q=0,\dots,n-2$.
Since $\cD_{ui}(1,1,\ldots,1)$ contains all unitary transformations of the form $(q\to p)$ for $q\le p$, all transformations $t_q=(q\to p_q)$ are present.
One verifies that applying $t_{n-2}t_{n-3}\cdots t_1t_0$ results in $t$.
Thus each non-decreasing transformation can be generated by  at most $n-1$ unitary transformations.  \qed
\end{proof}

 Note that there are semigroups with only non-decreasing transformations that do not have unitary generating sets;  each such semigroup, however, is a subsemigroup of $\mathcal{F}_Q$.
Since every $\mathcal{J}$-trivial language is also $\mathcal{R}$-trivial, the transition semigroups of all minimal DFAs accepting $\mathcal{J}$-trivial languages are also subsemigroups of $\cD_{ui}(1,1,\ldots,1)$.

\subsection{General Unitary Semigroups}
A~set $\{t_0,\dots,t_{k-1}\}$ of unitary transformations is \emph{$k$-cyclic} if  it has the form
$t_0=(q_0\to q_1)$, $t_1=(q_1\to q_2),\dots, t_{k-2}=(q_{k-2}\to q_{k-1})$, $t_{k-1}=(q_{k-1}\to q_0),
$
where the $q_i$ are distinct.

\begin{lemma} 
\label{lem:cycles}
Let $T$ be a set of unitary transformations.
\be
\item
If $T$ has a $k$-cyclic subset $\{t_0,\dots,t_{k-1}\}$ with $k\ge3$, then $\langle T \rangle$ is not aperiodic.
\item
If $T$ contains a subset $T_6=\{t_{01},t_{10},t_{12},t_{13},t_{21},t_{31}\}$ where 
$t_{i,j}=(q_i\to q_j)$ and 
$q_0,q_1,q_2,q_3\in Q$, then $\langle T \rangle$ is not aperiodic. 
\ee
\end{lemma}
\begin{proof}
Without loss of generality, we can replace $q_i$ by $i$ in both claims.
\be
\item
Suppose that $T$ contains $t_0,\dots,t_{k-1}$, where $k\le n$, $t_q=(q,q+1)$ for
$q=0,\dots, k-2$, and $t_{k-1}=(k-1\to 0)$.
Then
$t_{k-2}t_{k-3}\dots t_{1}t_{0} t_{k-1}$ maps 
 $0 \text{ to } 1, 1 \text{ to } 2,\dots, k-3 \text{ to } k-2, k-2\text{ to }  0, \text{ and } k-1 \text{ to } 0,$
 and does not affect any other states.
 Thus the set $\{0,1,\dots, k-2\}$ is cyclically permuted, which shows that $\langle T \rangle$ is not aperiodic.
\item
If  $\{t_{01},t_{12},t_{13},t_{10},t_{21},t_{31}\}\subseteq T$, then the transformation 
$t_{12}t_{01}t_{13}t_{21}t_{10}t_{31}$ transposes 0 and 1; hence $\langle T \rangle$ is not aperiodic.
\qed
\ee
\end{proof}

\begin{theorem}
If $\cD=(Q,\Sig,\delta,0,F)$ is unitary, the following are equivalent:
\be
\item
$\cD$ is aperiodic.
\item
The set of generators of $\cD$ does not contain any $k$-cyclic subsets with $k\ge 3$, and does not contain any sets of type $T_6$.
\item
Every strongly connected component of $\cD$ is a bipath.
\ee
\end{theorem}
\begin{proof}
 $1\Rightarrow 2\colon$ This follows from Lemma~\ref{lem:cycles}. 

$2\Rightarrow 3\colon$ Consider a strongly connected component $C$. If $|C|=1$, the claim holds.
Otherwise, suppose $p\in C$ and $(p\to q)$ is a transition. Then there must also be a directed path from $q$ to $p$. If the last transition in that path is $(r\to p)$, where $r \neq q$,
then the set of  generators must contain a  $k$-cyclic subset with $k\ge 3$, which is a contradiction.
Hence the transition $(q\to p)$ must be present. 

Next, suppose that there are transitions $(p\to q)$, $(p\to r)$, and $(p\to s)$.
By the argument above there must also be transitions $(q\to p)$, $(r\to p)$, and $(r\to s)$. 
But then the set of generators contains a subset of type $T_6$, which is again a contradiction.

It follows that every strongly connected component is a bipath, and the graph of the transitions of $\cD$ is a loop-free connection of such bipaths.

$3\Rightarrow 1\colon$
Since a bipath is monotonic, it is aperiodic by Proposition~\ref{prop:GoHo}.
By  Sch\"utzenberger's  theorem~\cite{Sch65},
the language of all words taking any state of the bipath to any other state of that bipath is star-free.
Since the graph of $\cD$ is a loop-free connection of bipaths, the language of all words taking any state of $\cD$ to any other state of $\cD$ is star-free.
Hence $\cD$ is aperiodic.
\qed
\end{proof}

A unitary DFA is \emph{complete} if the addition of any unitary transition results in a DFA that is not aperiodic.

\begin{theorem}A maximal aperiodic unitary semigroup is isomorphic to the transition semigroup of a complete unitary DFA $\mathcal{D}_u(n_1,\ldots,n_m)$,  where  $(n_1,\ldots,n_m)$  is some distribution of $Q$.
\end{theorem}
\begin{proof}
We know that an aperiodic unitary DFA $\cD$ is a loop-free connection of bipaths. Let $Q_1,\ldots,Q_m$ be the bipaths of $\cD$. There exists a linear ordering $<$ of them, such that there is no transformation $(p \to q)$ for $q \in Q_i,p \in Q_j,i<j$.  If  all possible transformations $(q \to p)$ for $q \in Q_i,p \in Q_j,i<j$ are present,  then $\cD$ is isomorphic to $\mathcal{D}_u(n_1,\ldots,n_m)$. Otherwise we can add more unitary  transformations  of Type 2 and obtain a larger semigroup.
\qed
\end{proof}

For each distribution $(n_1,\dots, n_m)$, we calculate the size of the transition semigroup of 
$\mathcal{D}_{ui}(n_1,\ldots,n_m)$.

\begin{theorem}\label{thm:unitary_size}
The cardinality of the transition semigroup of $\mathcal{D}_{ui}(n_1,\ldots,n_m)$ is
\begin{equation}
\prod_{i=1}^m \left(\binom{2n_i-1}{n_i} + \sum_{h=0}^{n_i-1} \left(\sum_{j=i+1}^{m} n_j\right)^{n_i-h} \binom{n_i}{h} \binom{n_i+h-1}{h} \right).
\end{equation}
\end{theorem}
\begin{proof}
As above $\cD_{ui}(n_1,\ldots,n_m)$ is a loop-free connection of bipaths, and its generators are the transformations within each bipath, all transformations of the form $(p\to q)$ where $p\in Q_i$, $q\in Q_j$, $i<j$, and $\one$.

In the transition semigroup of $\mathcal{D}_{ui}(n_1,\ldots,n_m)$, consider the transformation $t_i$ that (a) does not affect any states in $Q_j$ for $j<i$,  (b) maps some number $h$ of states of $Q_i$ to $Q_i$, and (c) maps the remaining $n_i-h$ states of $Q_i$ to some states in $Q_{i+1}\cup\dots\cup Q_m$. 
It is convenient to temporarily consider a partial transformation $t_i'$ which for all $q\in Q_i$ has the property $qt_i'=qt_i$, if $qt_i\in Q_i$ and $qt_i'=\Box$, otherwise. 
In other words, the images of the $n_i-h$ states mapped to the outside of $Q_i$ are all lumped together into the undefined value $\Box$. The number of such partial transformations generated by the transitions in the bipath is 
$\binom{n_i}{h} \binom{n_i+h-1}{h}$~\cite{GoHo92}; these are all the partially monotonic transformations of $Q_i$ that map exactly $h$ states of $Q_i$ to $Q_i$.

Returning now to $t_i$, consider first the case
$h=n_i$; then $t_i'$ is a total transformation equal to $t_i$, and there are $\binom{2n_i-1}{n_i}$ such transformations.
Otherwise, $t_i$ maps $n_i-h$ states of $Q_i$ to arbitrary states in $Q_{i+1}\cup\dots\cup Q_m$. If $k=n_{i+1}+\dots+n_m$ is the number of states in the bipaths below $Q_i$,
then for each $t_i'$ there are $k^{n_i-h}$  transformations $t_i$. 
Altogether,  for a fixed bipath $Q_i$, the number of transformations $t_i$ is
 \begin{equation}
 \binom{2n_i-1}{n_i} + \sum_{h=0}^{n_i-1} k^{n_i-h} \binom{n_i}{h} \binom{n_i+h-1}{h}.
 \end{equation} 

If $t$ is any transformation of $\mathcal{D}_{ui}(n_1,\ldots,n_m)$, then it can be represented by $t= t_m \circ t_{m-1} \circ \dots \circ t_1$, where $t_i$ maps  $Q_i$ into $Q_i \cup \ldots \cup Q_m$. Since the domains of $t_1,\dots,t_m$ are disjoint, there is a bijection between transformations $t$ and the sets $\{t_1,\dots,t_m\}$. Hence we can multiply the numbers of different transformations $t_i$ for each $1 \le i\le m$, and the formula in the theorem results. 
\qed
\end{proof}

Note that each factor of the product in Theorem~\ref{thm:unitary_size} depends only on $n_i$ and on the sum $k=n_{i+1}+\dots+n_m$. Hence if $\mathcal{D}_{ui}(n_1,\ldots,n_m)$ is maximal, then $\mathcal{D}_{ui}(n_2,\ldots,n_m)$ is also maximal and so on. 
Consequently, we have 

\begin{corollary}
Let $m_{ui}(n)$ be the cardinality of the largest transition semigroup of  DFA $\mathcal{D}_{ui}(n_1,\ldots,n_m)$ with $n$ states. If we define $m_{ui}(0) = 1$, then for $n>0$ 
{\small
\begin{equation}
m_{ui}(n) = \max_{j=1,\ldots,n} \left( m_{ui}(n-j) \left(\binom{2j-1}{j} + \sum_{h=0}^{j-1} (n-j)^{j-h} \binom{j}{h} \binom{j+h-1}{h} \right)\right).
\end{equation}
}
\end{corollary}

This leads directly to a dynamic algorithm taking $O(n^3)$ time for computing $m_{ui}(n)$ and the distributions  $(n_1,\ldots,n_m)$  yielding the maximal unitary semigroups.    
 This holds assuming constant time for computing the internal  terms  in the summation and summing them, where, however, the numbers can be very large ($O(n^n)$). The precise complexity depends on the algorithms used for multiplication, exponentiation and calculation of binomial coefficients. 

We were able to compute the maximal $\mathcal{D}_{ui}$ up to $n=1,000$. Here is an example of the maximal one for $n=100$:
$$\cD_{ui}(12,11,10,10,9,8,8,7,6,5,5,4,3,2);$$
its syntactic semigroup size exceeds $2.1 \times 10^{160}$. 
 Compare this to the previously known largest semigroup  of an $\mathcal{R}$-trivial language; its size is
$100!$ which is approximately $9.3 \times 10^{157}$.
On the other hand, the maximal possible syntactic semigroup of any regular language for $n=100$ is $10^{200}$.

\subsection{Asymptotic Lower Bound}

We were not able to compute  the  tight asymptotic bound on the maximal size of unitary semigroups. However, we computed  a  lower bound which is larger than $n!$, the previously known lower bound for the size of aperiodic semigroups.

\begin{theorem}
For $n$ even the size of the maximal unitary semigroup is at least 
$$ \frac{n!(n+1)!}{2^n ((n/2)!)^2}.$$
\end{theorem}
\begin{proof}
Let $n$ be even and  consider $\cD_{ui}(2,2,\dots,2)$ consisting of $m=n/2$ bipaths. From Theorem~\ref{thm:unitary_size} we have:
\begin{eqnarray*}
&&\prod_{i=1}^{m} \left(\binom{4-1}{2} + \sum_{h=0}^{1} \left(\sum_{j=i+1}^{m} 2\right)^{2-h} \binom{2}{h} \binom{2+h-1}{h} \right)\\
&=& \prod_{i=1}^{m} 4(m-i)^2 + 8(m-i) + 3 \\
&=& \prod_{i=1}^{m} (2i-1)(2i+1) \\
&=& (2m-1)!! (2m+1)!! \\
&=& (2m-1)!! (2(m+1)-1)!!
\end{eqnarray*}

By using the equality $(2k-1)!! = \frac{(2k)!}{2^k k!}$ we obtain:
\begin{eqnarray*}
&=& \frac{(2m)!}{2^m m!} \frac{(2(m+1))!}{2^{m+1} (m+1)!} \\
&=& \frac{(2m)!(2m+2)!}{2^{2m+1} m! (m+1)!} \\
&=& \frac{n!(n+2)!}{2^{n+1} (n/2)! ((n/2)+1)!} \\
&=& \frac{n!(n+2)(n+1)!}{2^{n+1} (n/2)! (n/2+1)(n/2)!} \\
&=& \frac{n!(n+1)!}{2^n ((n/2)!)^2}.
\end{eqnarray*}
\qed
\end{proof}

For $n=100$ the bound exceeds $7.5 \times 10^{158}$.  Larger lower bounds can also be found using increasing values of $j$ in $\cD_{ui}(j,j,\dots,j)$, but the complexity of the calculations increases, and such bounds are not tight. 

\goodbreak
\section{Semiconstant Semigroups}
\label{sec:semi2}
 We now consider our largest aperiodic semigroups, the semiconstant ones.

\subsection{Nearly Monotonic Semigroups}

Let $K_Q$ be the set of all constant transformations of $Q$, and $NM_Q = PM_Q \cup K_Q$. 
We call the transformations in $NM_Q$ \emph{nearly monotonic} with respect to the usual order on integers.
The next result follows from Proposition~\ref{prop:pmon} and~\cite{BLL13}.
 \begin{proposition} 
 \label{prop:nmon}
Let $n\ge 2$ and $\cD_{scti}((n-1,1))=(Q, \Sig_{scti},\delta_{scti},0,\{n-1\})$.  Then
 \be
 \item 
 Each of the $3n-4$ transformations of $\cD_{scti}((n-1,1))$ of Types 1, 2, and 4 is partially monotonic, and there is one constant transformation $(Q\to 0)$.
 Thus DFA $\cD_{scti}((n-1,1))$ is nearly monotonic, and hence aperiodic.
 \item
 The transition semigroup $NM_Q$ of $\cD_{scti}((n-1,1))$ consists of all the $h(n)$ nearly monotonic transformations of $Q$, where
 \begin{equation}
h(n)=e(n)+n-1.
\end{equation}
 \item
 Each generator, other than the constant and $\one$, is unitary, and $3n-3$ is the smallest number of unitary, constant and identity generators of $NM_Q$.
 \ee
 \end{proposition}

For $n\ge 4$ the semigroup of all nearly monotonic transformations is larger than the semigroup of all partially monotonic transformations.
 Note that there are nearly monotonic semigroups that do not have semiconstant generating sets;  each nearly monotonic semigroup, however, is a subsemigroup of the transition semigroup of $\cD_{scti}((n-1,1))$.

\subsection{Semiconstant Tree Semigroups}

An example of a maximal semiconstant tree DFA for $n=6$ is $\cD_{scti}(((2,2),2))$; its transition semigroup has 1,849 elements.
For $n\ge 4$, the maximal semiconstant tree semigroup is the largest aperiodic semigroup known.

First we define a new operation on DFAs. 

\begin{definition}
Let $\mathcal{A}=(Q_\cA,\Sig_\cA,\delta_\cA,q_\cA,F_\cA)$ and $\mathcal{B}=(Q_\cB,\Sig_\cB,\delta_\cB,q_\cB,F_\cB)$ be DFAs. Let $Q_\cC=Q_A\cup Q_B$.
The \emph{semiconstant sum} of $\mathcal{A}$ and $\mathcal{B}$ is the DFA 
$\mathcal{C}=(\mathcal{A},\mathcal{B})=(Q_\cC,\Sig_\cC,\delta_\cC,q_\cA,F_\cB)$.
For each transition $t$ in $\delta_A$, we have a transition $t'$ in $\delta_C$ such that $qt'=qt$ for $q \in Q_\mathcal{A}$ and $qt'=q$ otherwise. 
Dually,  we have transitions defined by $t$ in $\delta_B$. Moreover we have a unitary transformation $(p \to q)$ for each $p \in Q_\mathcal{A},q \in Q_\mathcal{B}$, and a constant transformation $(Q_\mathcal{C} \to q_\mathcal{A})$.

\end{definition}

 For $m>1$, each $\cD_{scti}(\Delta_Q(n_1,\dots,n_m))$ is a semiconstant sum of two smaller semiconstant tree DFAs $\cD_{scti}(\Delta_{Q_{left}}(n_1,\dots,n_r))$, defined by the left subtree of the root of $\Delta_Q(n_1,\dots,n_m)$, and $\cD_{scti}(\Delta_{Q_{right}}(n_{r+1},\dots,n_m))$, defined by the right subtree.

\begin{lemma}\label{lem:scsum_minimal}  
The semiconstant sum $\cC = (\cA,\cB)$ is minimal if and only if every state of 
$\cA$ is reachable from $q_\cA$, the states of $\cB$ are pairwise distinguishable, and $F_\cB$ is non-empty.
\end{lemma}
\begin{proof}
If $\cC$ is minimal, then every state of $\cC$ is reachable from $q_\cA$ in $\cC$. Since any transformation mapping a state $q \in Q_\cB$ to a state from $Q_\cA$ is composed from the constant transformation $(Q_\cC \to q_\cA)$, every state of $\cA$ is reachable from $q_\cA$.
Now consider two distinct states $p,q \in Q_\cB$. 
Since $\cC$ is minimal, $p$ and $q$ are distinguishable by some word $w$, and no letter of $w$ can induce the constant transformation $(Q_\cC \to q_\cA)$.
Hence every letter of $w$ induces a transformation that acts on $Q_\cB$ either as the identity or as some $t' \in \delta_\cB$. 
If we omit the letters that act as the identity, we obtain a word $w'$ that distinguishes $p$ and $q$ in $\cB$.


Conversely, distinct states $p \in Q_\cA$, $q \in Q_\cC \setminus F_\cB$ are distinguishable as follows.
Apply a unitary transformation $t$ that takes $p$ to a state in $F_\cB$. Since $q$ is not changed by $t$, $p$ and $q$ are distinguishable. If $p \in Q_\cA$ and $q \in F_\cB$ then $p$ and $q$ are already distinguished (by the empty word). If $p \in Q_\cB$ and $q \in Q_\cB$ then they are distinguishable by assumption.
Every state of $\cA$ is reachable from $q_\cA$ by assumption. Also, any state in $q \in Q_\cB$ is reachable from $q_\cA$ by a unitary transformation. Hence all the states of $\cC$ are reachable, and $\cC$ is minimal.
\qed
\end{proof}

\begin{lemma}\label{lem:scsum_aperiodic}
 If $\mathcal{A}$ and $\mathcal{B}$ are aperiodic, then their semiconstant sum  $(\mathcal{A},\mathcal{B})$ is also aperiodic. 
\end{lemma}
\begin{proof}
 Suppose that  $\langle (\mathcal{A},\mathcal{B}) \rangle$ contains a cycle $t$. 
This cycle cannot include both a state from $\cA$ and a state from $\cB$, since the only way to map a state from $\cB$ to a state from $\cA$  in $(\cA,\cB)$ is by a constant transformation, and a constant transformation cannot be used as a generator of a cycle.
Hence all the cyclic states must be either in $Q_\cA$ or $Q_\cB$, which contradicts the assumption that $\mathcal{A}$ and $\mathcal{B}$ are aperiodic.
\qed
\end{proof}

An DFA is \emph{transition-complete} if it is aperiodic  and adding any transition to it destroys aperiodicity.

\begin{lemma}\label{lem:scsum_no_more_transitions}
If $\mathcal{A}$ and $\mathcal{B}$ are  transition-complete, their semiconstant sum $(\mathcal{A},\mathcal{B})$ is also transition-complete.
\end{lemma}
\begin{proof}  
We know from Lemma~\ref{lem:scsum_aperiodic} that $(\cA,\cB)$ is aperiodic.
Suppose that a new transformation $t$ can be added to $(\cA,\cB)$ in such away that the resulting DFA remains aperiodic.  We consider the following cases depending on the image $Q_\cB t$.

If $|Q_\cB t \cap Q_\cA| = 0$ then $t=t_\cB \circ u \circ t_\cA$, where $t_\cA$ and $t_\cB$ are transformations changing only the states of $Q_\cA$ and $Q_\cB$, respectively, and $u$ only maps some of the states of $Q_\cA$ to $Q_\cB$. 
If $t$ is new, then one of $t_\cB$, or $t_\cA$  or $u$ is new. But  we know that no new transition can be added to $\cA$ or $\cB$, and we have all possible transitions of type $u$.

If  $|Q_\cB t \cap Q_\cA| > 0$ and  $|Q_\cC t| = 1$, then  $t$ is a constant transformation that we have already, because we have $(Q_\cB \to q_\cA),q_\cA \in Q_\cA$  from the construction of semiconstant sum, and each constant transformation on $Q_\cA$, since $\cA$ is transition-complete.

If $|Q_\cB t \cap Q_\cA| > 0$ and  $|Q_\cC t| > 1$, then let $q_1 \in Q_\cB$ be some state such that $q_1t=p_1 \in Q_\cA$, and let  $q_2 \in Q_\cC$  be some state such that $q_2t=p_2$, where $q_1 \neq q_2$ and $p_2 \neq p_1$.

If $p_2 \in Q_\cB$, let $c = (Q_\cB \to q_1)$; since we cannot add any transformation to $\cB$ and $c$ is constant in $\cB$, it must be present. Otherwise, let $c=(p_2 \to q_1)$. Note that $c$ does not affect $p_1$. Similarly, if $q_2 \in Q_\cA$, let $d = (Q_\cA \to q_2)$; otherwise, let $d = (p_1 \to q_2)$. Note that $d$ does not affect $q_1$. 
Then the transformation $t'=t\circ c\circ d$ is such that $q_1t' =q_2$ and $q_2t'=q_1$ and the DFA cannot be aperiodic. 
\end{proof}

\begin{corollary}
All semiconstant tree DFAs of the form $\cD_{scti}(\Delta_Q(n_1,\dots,n_m))$ are transition-complete.
\end{corollary}
\begin{proof}
This follows by induction on $m$. For $m=1$ we have a bipath. For $m>1$, $\cD_{scti}(\Delta_Q(n_1,\dots,n_m))$ is the semiconstant sum of $\cD_{scti}(\Delta_Q(n_1,\dots,n_r))$ and $\cD_{scti}(\Delta_Q(n_{r+1},\dots,n_m))$. By the inductive assumption and 
Lemma~\ref{lem:scsum_aperiodic} we know that $\cD_{scti}(\Delta_Q(n_1,\dots,n_m))$ is also aperiodic. And by Lemma~\ref{lem:scsum_no_more_transitions} we know that no more transitions can be added.
\qed
\end{proof}

In order to count the size of the semigroup of a semiconstant sum, we extend the concept of partial transformations to $k$-partial transformations.

\begin{definition}A \emph{$k$-partial transformation} of $Q$ is a transformation of $Q$ into $Q \cup \{\Box_1,\Box_2,\dots,\Box_k\}$, where $\Box_1,\Box_2,\dots,\Box_k$ are pairwise distinct, and distinct from all $q \in Q$.\end{definition}

Let $\cA=(Q,\Sig,\delta,s,F)$ be a DFA, and let $t$ be a $k$-partial transformation of $Q$. We say that $t$ is \emph{consistent} for $\cA$ if there exists $t'$ in $\delta$ such that
 if $qt \in Q$, then $qt = qt'$ for all $q \in Q$.

The set of consistent $k$-partial transformations of a semigroup describes its  potential for  forming a large number  of transformations, when used in a semiconstant sum. For a fixed $n \ge 6$, there exist semigroups with smaller cardinalities than the maximal ones, but with larger numbers of consistent $k$-partial transformations for some $k$. Thus $k$-partial transformations are useful for finding such non-maximal semigroups, as they can result in larger semigroups when used in sums.

The transition semigroup of $\cA$ can be characterized by a function $f_\cA\colon \mathbb{N} \to \mathbb{N}$ counting all consistent $k$-partial transformations for a given $k$.
For example, for $k=1$, $f_\cA$ is the number of all consistent partial transformations for $\cA$. For a DFA $\cA = \cD_{ui}(n_1,\dots,n_m)$, $f_\cA(1)$ is the size of the semigroup of $\cD_{ui}(n_1,\dots,n_m,1)$.

From the proof of Theorem~\ref{thm:unitary_size} we know that the number of consistent $k$-partial transformations for a bipath of size $n$  having an identity transformation  is
$$m_{bi}(n,k) = \binom{2n-1}{n} + \sum_{h=0}^{n-1} k^{n-h} \binom{n}{h} \binom{n+h-1}{h}.$$

\begin{theorem}\label{thm:composed_k_partial_transformations}
Let $\cA$ and $\cB$ be strongly connected DFAs with $n$ and $m$ states, respectively. Let $f_\cA(k)$ and $f_\cB(k)$ be the functions counting their consistent $k$-partial transformations. Then the function $f_\cC$ counting the consistent $k$-partial transformations of the semiconstant sum $\cC=(\mathcal{A},\mathcal{B})$ is
$$f_\cC(k) = f_\cA(m+k)f_\cB(k) + n(k+1)^n ((k+1)^m - k^m).$$
\end{theorem}
\begin{proof}
Consider a transformation $t$ in $\langle (\mathcal{A},\mathcal{B}) \rangle$.

If $|Q_\cB t \cap Q_\cA| = 0$, then $t = t_\cB \circ t_\cA$, where $t_\cB$ is a $k$-partial transformation of $Q_B$ into $Q_B$, and $t_A$ is a $k$-partial transformation of $Q_\cA$ into $Q_\cB$. Moreover $t_\cB$ and $t_\cA$ are uniquely defined by the images $Q_\cB t$ and $Q_\cA t$, respectively.
We have $f_\cB(k)$ possible $t_B$ transformations, and $f_\cA(m+k)$ possible $t_\cA$ transformations, since $t_\cA$ corresponds to an $(m+k)$-partial transformation of $\cA$. So we have $f_\cA(m+k)f_\cB(k)$ different $k$-partial transformations $t$ in this case.

If $|Q_\cB t \cap Q_\cA| \le 1$, then the constant generator $c = (Q_\cB \to q_\cA), q_\cA \in Q_A$ must be used, since it is the only generator mapping a state from $Q_\cB$ into a state from $Q_\cA$. So the case $|Q_\cB t \cap Q_\cA| > 1$ is not possible.
 For each state $q \in Q_\cA$ either $qt$ is one of the $k$ undefined values or $qt=q_\cA$. This yields $k+1$ possible mappings for a given $x$, and $(k+1)^n$ possibilities in total.
Also, for each state $q \in Q_\cB$ either $qt$ is one of the $k$ undefined values or $qt=q_\cA$. However,  the latter case must occur for at least one $q \in Q_\cB$. This yields $(k+1)^m - k^m$ possibilities in total.
Because $\cA$ is strongly connected, we have $n$ possibilities for the selection of $q$. This yields $n (k+1)^n ((k+1)^m - k^m)$ different $k$-partial transformations in this case.

Altogether, we have $f_\cA(m+k)f_\cB(k) + n(k+1)^n ((k+1)^m - k^m)$.
\qed
\end{proof}

\begin{corollary}
The number of $k$-partial transformations of $\cD_{scti}(\Delta_Q(n_1,\dots,n_m))$ of size $n$ is:
{\footnotesize
\[f_\cD(k) = \left\{ \begin{array}{ll}
			m_{bi}(n,k), & \mbox{if $m=1$};\\
			f_{\cD_{left}}(r+k)f_{\cD_{right}}(k) + \ell(k+1)^{\ell} ((k+1)^{r} - k^{r}), & \mbox{if $m>1$},
			\end{array}
		\right. \]
		}

\noindent where $\cD_{left}$ is the DFA defined by $\Delta_{Q_{left}}(n_1,\dots,n_i)$,
the left subtree of the tree $\Delta_Q(n_1,\dots,n_m)$,
$\cD_{right}$ is defined by $\Delta_{Q_{right}}(n_{i+1},\dots,n_m)$, 
the right subtree of $\Delta_Q(n_1,\dots,n_m)$, and $\ell$, $r$ are the numbers of states in $\cD_{left}$ and $\cD_{right}$, respectively.
\end{corollary}
\begin{proof}
This follows from Theorems~\ref{thm:unitary_size} and~\ref{thm:composed_k_partial_transformations}.
\qed
\end{proof}

The size of the semigroup of DFA $\cD_{scti}(\Delta_Q(n_1,\dots,n_m))$ is $f_\cD(0)$.

\begin{corollary}
Let $m_{scti}(n,k)$ be the maximal number of $k$-partial transformations of a semiconstant DFA $\cD_{scti}(\Delta_Q(n_1,\dots,n_m))$ with $n$ states. Then
\begin{equation}
m_{scti}(n,k) = \max\begin{cases}
m_{bi}(n,k) \\
\max\limits_{s=1,\dots,n-1}\left\{\begin{aligned}
& m_{scti}(n-s,s+k)m_{scti}(s,k) \\
& + (n-s)(k+1)^{n-s} ((k+1)^{s} - k^{s})\end{aligned}
\right\}.
\end{cases}
\end{equation}
\end{corollary}
\begin{proof}
A semiconstant tree DFA is either a bipath, or a semiconstant sum of two smaller semiconstant tree DFAs. Since its number of transformations depends only on the numbers of $k$-partial transformations of the smaller ones, we can use the maximal ones, and select the best split for the sum.
\end{proof}

The maximal size of semigroups of the DFAs $\cD_{scti}$ with $n$ states is $m_{scti}(n,0)$.

Instead  of a bipath and the value $m_{bi}(n,k)$ we could use any strongly connected automaton with an aperiodic semigroup. If such a semigroup would have a larger number of $k$-partial transformations  than our semiconstant tree DFAs for some $k$, then we could obtain even larger aperiodic semigroups. 

The corollary  results directly in a dynamic algorithm working in $O(n^3)$ time (assuming constant time for arithmetic operations and computing binomials) for computing $m_{scti}(n,0)$, and the distribution with the full binary tree yielding the maximal semiconstant tree semigroup. 

We computed the maximal semiconstant tree semigroups up to $n=500$. For $n=100$, for example, one of the maximal DFAs is
\begin{eqnarray*}
\cD_{scti} & & (((((((2,2),(2,2)),((2,2),(2,2))),(((2,2),(2,2)),((2,2),3))),\\
& & \ ((((2,2),3),(3,3)),((3,3),(3,3)))),((((3,2),(3,2)),((3,2),(2,2))),\\
& & \ ((2,2),(2,2)))),(((3,3),(3,2)),((2,2),2))),
\end{eqnarray*}
and its syntactic semigroup size exceeds $3.3 \times 10^{160}$.

\section{Quotient Complexity of Reversal and Product}
\label{sec:rev}

 In this section we refer to quotient/state complexity simply as complexity. We settle two open problems about the  complexity of reversal and concatenation of star-free languages.

\subsection{Reversal}
The bound $2^n-1$ is reachable by the reverse of a star-free language~\cite{BrLiu12}, but it was not known if $2^n$ can be reached. We answer this question now.

A~\emph{nondeterministic finite automaton (NFA)} is a quintuple 
$\cN=(Q, \Sig, \delta, I,F)$, where 
$Q$, $\Sig$, and $F$ are as in a DFA,  $\delta\colon Q\times \Sig\to 2^Q$ is the  \emph{transition function}, and
$I\subseteq  Q$ is the   \emph{set of initial states}.
As usual, we extend $\delta$ to functions 
$\delta\colon Q\times \Sig^*\to 2^Q$, and $\delta\colon 2^Q\times \Sig^*\to 2^Q$.
The \emph{language accepted} by an NFA $\cN$ is 
$L(\cN)=\{w\in\Sig^*\mid \delta(I,w)\cap F\neq \emp\}$.

Consider a DFA $\cD=(Q, \Sig, \delta, 0,F)$ accepting a language $L$. 
We extend $\delta$ to a function $\delta\colon 2^Q \to 2^Q$ as usual. 
To construct a DFA for $L^R$, we first take the NFA $\cD^R=(Q, \Sig, \delta^R, F,\{0\})$, where $q\in \delta^R(p,a)$ if and only if $\delta(q,a)=p$ for all $p,q\in Q$, $a\in \Sig$; here $\delta^R$ is a relation.
We then determinize $\cD^R$ by the subset construction to get the DFA $\cD^{RD}=(2^Q, \Sig, \delta^R, F,\{0\})$, where now $\delta^R$ is a function $\delta^R\colon 2^Q\to 2^Q$.
For any $P\subseteq Q$, let $\ol{P}=Q\setminus P$.
\begin{lemma}
\label{lem:rev}
For any $P \subseteq Q$ and $w \in \Sigma^*$, $\delta^R(\ol{P},w) =\ol{\delta^R(P,w)}$.
\end{lemma}
\begin{proof}
By the reversal and subset constructions, for $P,S\subseteq Q$, we have $\delta^R(S,w)=P$  if and only if $\delta(P,w^R)=S$. Thus
$$\ol{\delta^R(P,w)}=\{q\in Q \mid \delta(q,w^R)\not\in P\}=\{q\in Q \mid \delta(q,w^R)\in \ol{P}\}=\delta^R(\ol{P},w). \hspace{.8cm} \qed $$
\end{proof}

\begin{theorem}
\label{thm:rev}
If $L\subseteq \Sig^*$ is star-free, the  complexity of $L^R$ is at most $2^n-1$.
\end{theorem}
\begin{proof}
Let $\cD$ be the DFA accepting $L$.
In the DFA $\cD^{RD}$, state $\ol{F}$ must be reachable from state $F$, if all $2^n$ states are reachable.
If  $\delta^R(F,w) = \ol{F}$, then by Lemma~\ref{lem:rev},
$\delta^R(\ol{F},w) = F$. 
Hence  $\delta^R(F,ww)=F$ and $\delta^R(\ol{F},ww) = \ol{F}$.
So $ww$ transposes $F$ and $\ol{F}$ in $\cD^{RD}$.
Therefore $\cD^{RD}$ is not aperiodic, and $L^R$ and $L$ are not star-free.
\qed
\end{proof} 

\subsection{Product}

The  complexity of product (catenation, concatenation) of star-free languages has the same tight upper bound as the product of regular languages, with a small exception~\cite{BrLiu12}.
Let $K$ and $L$ be star-free languages with  complexities $m$ and $n$, respectively.
If $m\ge 1$ and $n\ge 3$, the  complexity of $KL$ is $(m-1)2^n+2^{n-1}$, which is the  bound for regular languages.
If $m=1$, then either $K=\emp$ and the complexity of $KL=\emp$ is 1, 
or $K=\Sig^*$. In the latter case, the  complexity of $\Sig^*L$ is $2^{n-1}$, which is again the same as  for regular languages~\cite{BrLiu12}.
If $n=1$, then either $KL=\emp$, or $KL=K\Sig^*$. In the second case the bound is $m$ for both star-free and regular languages~\cite{BJL13,YZS94}.

This leaves the case of $m \ge2 $ and $n=2$.
Let $\cD_K=(Q_K, \Sig, \delta_K, q_0,F_K)$, where $Q_K=\{q_0,q_1,\dots,q_{m-1}\}$, and let 
$\cD_L=(\{0,1\}, \Sig, \delta_L, 0,F_L)$ be the DFAs accepting $K$ and $L$, respectively.
It was shown in~\cite{BrLiu12} that if $F_L=\{0\}$, then $3m-2$ is a tight upper bound on the  complexity of $KL$; otherwise, $3m-1$ is an upper bound.
Here we prove that $3m-2$ is also a tight upper bound if $F_L=\{1\}$, thus closing the gap.

Assume from now on that $F_L=\{1\}$. The state complexity of $KL$ is maximized if $\cD_K$ has only one final state~\cite{YZS94}; without loss of generality we can assume that $F_K=\{q_{m-1}\}$.
To find the DFA accepting $KL$ we first construct an $\eps$-NFA  $\cN$ (an NFA with empty word transitions)  from $\cD_K$ and $\cD_L$ by adding an $\eps$-transition from the final state $q_{m-1}$ of $\cD_K$ to the initial state 0 of $\cD_L$. The set of states of $\cN$ is $Q_K\cup \{0,1\}$, the set of initial states  is $\{q_0\}$, and the set of final states is~$\{1\}$.

Next, we the use the subset construction on $\cN$ to find the DFA 
$\cD_{KL}=(Q,\Sig,\delta,\{q_0\},\{\{1\}\})$ accepting $KL$.
The states of $\cD_{KL}$ are subsets of $Q_K\cup \{0,1\}$.
 The possibly reachable sets of states of $\cN$ are the following: the $(m-1)\cdot 2^2$ sets of the form $\{q_i\}\cup S$, where $i\neq m-1$, $S\in \{\emp,\{1\},\{2\},\{1,2\}\}$, and the two sets $\{q_{m-1},0\}$ and 
$\{q_{m-1},0,1\}$, for a total of at most $4m-2$ states.

\begin{lemma}
\label{lem:prod}
If $F_L=\{1\}$, then for $i=1,\dots, m-1$, $\{q_i,1\}$ is equivalent to  $\{q_i,0,1\}$.
Also,  for any $i,j \neq m-1$ and $i \neq j$, $\{q_i,1\}$ and $\{q_j,1\}$ are equivalent.

\end{lemma}
\begin{proof}
As noted in~\cite{BrLiu12}, 
$\{q_i,1\}$ and $\{q_i,0,1\}$ are equivalent for all $i\neq m-1$, because only three types of transitions are possible in $\cD_L$: (a) the identity transition, (b) the constant transition $(\{0,1\}\to 0)$, and (c) the constant transition $(\{0,1\}\to 1)$.
Thus we can keep the representative  $\{q_i,1\}$ and remove $\{q_i,0,1\}$.
Hence there are at most $4m-2-(m-1)=3m-1$ states.

For the second claim, assume to the contrary that $\{q_i,1\}$ and $\{q_j,1\}$ are distinguishable in $\cD_{KL}$. 
Then there exists some $w \in \Sigma^*$ such that (without loss of generality) $\delta(\{q_i,
1\},w)$ contains $1$ ($w$ is accepted from $\{q_i,1\}$) and $\delta(\{q_j,1\},w)$ does not contain $1$ ($w$ is rejected from $\{q_j,1\}$).

Since $\delta(\{q_j,1\},w)$ does not contain $1$, we have $\delta_L(1,w)=0$. Since $\delta(\{q_i,1\},w)$ contains $1$, also $\delta(\{q_i\},w)$ must contain $1$.
Since the only way to reach $1$ from $q_i$ is through $0$, there are some words $u$ and $v$ such that $uv = w$, $\delta(\{q_i\},u)$ contains~$0$ and $\delta_L(0,v)=1$. We have two cases depending on $\delta_L(1,u)$:
\be
\item
If $\delta_L(1,u)=0$, then $\delta_L(1,uv)=1$, contradicting  that $\delta_L(1,w)=0$.
\item
If $\delta_L(1,u)=1$, then $\delta_L(1,v)=0$ (because $\delta_L(1,w)=0$). But we know that $\delta_L(0,v)=1$; so $v$ induces the  cycle $(0,1)$ and $L$ is not aperiodic.
\ee
Thus all $m-1$ states of the form $\{q_i,1\}$ can be represented by one state, say $\{q_0,1\}$, reducing the possible number of indistinguishable states of $\cD_{KL}$ to $3m-1-(m-2)=2m+1$.
Therefore in the case where $F_L=\{1\}$, the state complexity of $KL$ is at most $2m+1$, which is smaller than $3m-2$. 
\qed
\end{proof}
As a consequence of the lemma, the result of~\cite{BrLiu12} can be strengthened as follows:
\begin{theorem}
\label{thm:product}
Let $K$ and $L$ be star-free languages with complexities $m\ge 2$ and 2, respectively;
then the complexity of $KL$ is $3m-2$. 
\end{theorem}

\section{Conclusions}
\label{sec:conc}

We have found two new types of aperiodic semigroups. Maximal semiconstant semigroups of type $\cD_{scti}(\Delta_Q(n_1,\dots,n_m))$ are currently the largest aperiodic semigroups known. A~tight upper bound on the size of aperiodic semigroups remains unknown.


\begin{thebibliography}{10}
\providecommand{\url}[1]{\texttt{#1}}
\providecommand{\urlprefix}{URL }

\bibitem{Brz10}
Brzozowski, J.: Quotient complexity of regular languages. J. Autom. Lang. Comb.
   15(1/2),  71--89 (2010)

\bibitem{Brz12a}
Brzozowski, J.: In search of the most complex regular languages. Internat.\ J.\
  Found.\ Comput.\ Sci.  24(6),  691--708 (2013)

\bibitem{BrFi80}
Brzozowski, J., Fich, F.E.: Languages of {${\mathcal R}\/$}-trivial monoids. J.
  Comput. System Sci.  20(1),  32--49 (1980)

\bibitem{BJL13}
Brzozowski, J., Jir{\'a}skov{\'a}, G., Li, B.: Quotient complexity of ideal
  languages. Theoret. Comput. Sci.  470,  36--52 (2013)

\bibitem{BrLi13}
Brzozowski, J., Li, B.: Syntactic complexity of \mbox{$\mathcal R$}- and
  \mbox{$\mathcal J$}-trivial languages. In: J\"urgensen, H., Reis, R. (eds.)
  DCFS 2013. LNCS, vol. 8031, pp. 160--171. Springer (2013)

\bibitem{BLL13}
Brzozowski, J., Li, B., Liu, D.: Syntactic complexities of six classes of
  star-free languages. J. Autom. Lang. Comb.  17(2--4),  83--105 (2012)

\bibitem{BrLiu12}
Brzozowski, J., Liu, B.: Quotient complexity of star-free languages. Internat.\
  J.\ Found.\ Comput.\ Sci.  23(6),  1261--1276 (2012)

\bibitem{Die10}
Diestel, R.: Graph Theory, Graduate Texts in Mathematics, vol. 173.
  Springer-Verlag, Heldelberg, fourth edn. (2010),
  http://diestel-graph-theory.com

\bibitem{Gin66}
Ginzburg, A.: Abour some properties of definite, reverse definite and related
  automata. IEEE Trans. Electronic Comput.  EC--15,  806--810 (1966)

\bibitem{GoHo92}
Gomes, G., Howie, J.: On the ranks of certain semigroups of order-preserving
  transformations. Semigroup Forum  45,  272--282 (1992)

\bibitem{How71}
Howie, J.M.: Products of idempotents in certain semigroups of transformations.
  Proc. Edinburgh Math. Soc.  17(2),  223--236 (1971)

\bibitem{IvNG14}
Iv\'an, S., Nagy-Gy\"orgy, J.: On nonpermutational transformation semigroups
  with an application to syntactic complexity (2014),
  \url{http://arxiv.org/abs/1402.7289}

\bibitem{KiSz13}
Kisielewicz, A., Szyku{\l}a, M.: Generating small automata and the \v{C}ern\'y
  conjecture. In: Konstantinidis, S. (ed.) CIAA 2013. LNCS, vol. 7982, pp.
  340--348. Springer (2013)

\bibitem{McNP71}
McNaughton, R., Papert, S.A.: Counter-Free Automata, MIT Research Monographs,
  vol.~65. The MIT Press (1971)

\bibitem{Pin97}
Pin, J.E.: Syntactic semigroups. In: Handbook of Formal Languages, vol.~1:
  Word, Language, Grammar, pp. 679--746. Springer, New York, NY, USA (1997)

\bibitem{Sch65}
Sch\"utzenberger, M.: On finite monoids having only trivial subgroups. Inform.
  and Control  8,  190--194 (1965)

\bibitem{Yu01}
Yu, S.: State complexity of regular languages. J. Autom. Lang. Comb.  6,
  221--234 (2001)

\bibitem{YZS94}
Yu, S., Zhuang, Q., Salomaa, K.: The state complexities of some basic
  operations on regular languages. Theoret. Comput. Sci.  125,  315--328 (1994)

\end{thebibliography}
\providecommand{\noopsort}[1]{}

\end{document}